\documentclass[conference]{IEEEtran}

\usepackage{amsfonts,amsbsy,bm,amsmath}
\usepackage[nolist]{acronym}
\usepackage{mathrsfs} 
\usepackage{graphicx,cite,amssymb,amsmath,mathtools,amsthm}
\usepackage{tikz}
\usetikzlibrary{fadings}
\usetikzlibrary{shadows.blur}
\usetikzlibrary{shapes,arrows}
\usetikzlibrary{calc}
\usetikzlibrary{decorations.pathreplacing}
\usepackage{ifthen}
\usetikzlibrary{positioning}

\usepackage{pgfplots,relsize}
\usepackage{color}

\usepackage{cleveref}
\usepackage{bbm}
\allowdisplaybreaks

\DeclareRobustCommand\mytikzVN{\tikz \node[circle,fill=white,draw,radius=0.2cm,inner sep=0.05cm] at (0,0) {\tiny{$=$}};}

\DeclareRobustCommand\mytikzCN{\tikz \node[circle,fill=white,draw,radius=0.2cm,inner sep=0.05cm] at (0,0) {\tiny{$+$}};}

\DeclareMathOperator\atanh{atanh}

\begin{document}
	\newcommand{\code}{\mathcal{C}}
\newcommand{\vecu}{\mathbf{u}}
\newcommand{\vecx}{\mathbf{x}}
\newcommand{\vecuhat}{\hat{\mathbf{u}}}
\newcommand{\vecc}{\mathbf{c}}
\newcommand{\vecchat}{\hat{\mathbf{c}}}
\newcommand{\vecy}{\mathbf{y}}
\newcommand{\G}{\mathbf{G}}
\newcommand{\GH}[1]{\bm{\mathsf{G}}_{#1}}
\newcommand{\T}[1]{\bm{\mathsf{T}}{(#1)}}

\newcommand{\mi}{\mathrm{I}}
\newcommand{\Prob}{\mathbf{P}}
\newcommand{\Z}{\mathrm{Z}}
\newcommand{\W}{\mathrm{W}}
\newcommand{\SC}{\mathrm{SC}}
\newcommand{\MAP}{\mathrm{MAP}}
\newcommand{\E}{\mathrm{E}}
\newcommand{\U}{\mathrm{U}}

\newcommand{\Amin}{\mathrm{A}_{\mathrm{min}}}
\newcommand{\dmin}{d}
\newcommand{\erfc}{\mathrm{erfc}}

\newcommand{\de}{\mathrm{d}}

\newcommand{\oleq}[1]{\overset{\text{(#1)}}{\leq}}
\newcommand{\oeq}[1]{\overset{\text{(#1)}}{=}}
\newcommand{\ogeq}[1]{\overset{\text{(#1)}}{\geq}}
\newcommand{\ogeql}[2]{\overset{#1}{\underset{#2}{\gtreqless}}}

\newcommand{\pscd}{\Prob(\mathcal{E}_{\SCD})}
\newcommand{\uED}{\hat{u}}
\newcommand{\LED}{L_i^{(\ED)}}
\newcommand{\inverson}[1]{\mathbb{I}\left\{#1\right\}}

\newtheorem{mydef}{Definition}
\newtheorem{prop}{Proposition}
\newtheorem{theorem}{Theorem}

\newtheorem{lemma}{Lemma}
\newtheorem{example}{Example}
\newtheorem{corollary}{Corollary}

\renewcommand{\qed}{\hfill\ensuremath{\blacksquare}}

\begin{acronym}
	\acro{BCJR}{Bahl, Cocke, Jelinek, and Raviv}
	\acro{BEC}{binary erasure channel}
	\acro{LDPC}{low-density parity-check}
	\acro{PC}{product code}
	\acro{LLR}{log-likelihood ratio}
	\acro{BLEP}{block error probability}
	\acro{B-DMC}{binary-input discrete memoryless channel}
	\acro{BSC}{binary symmetric channel}
	\acro{MAP}{maximum a posteriori}
	\acro{ML}{maximum likelihood}
	\acro{SPC}{single parity-check}
	\acro{SC}{successive cancellation}
	\acro{VN}{variable node}
	\acro{CN}{check node}
	\acro{RM}{Reed-Muller}
	\acro{RV}{random variable}
	\acro{SCC}{super component codes}
	\acro{B-AWGN}{binary input additive white Gaussian noise}
	\acro{CER}{codeword error rate}
	\acro{BER}{bit error rate}
\end{acronym}

	\title{Successive Cancellation Decoding of\\Single Parity-Check Product Codes}

	\author{\IEEEauthorblockN{Mustafa Cemil Co\c{s}kun\IEEEauthorrefmark{1},
	Gianluigi Liva\IEEEauthorrefmark{1},
	Alexandre Graell i Amat\IEEEauthorrefmark{2} and 
	Michael Lentmaier\IEEEauthorrefmark{3}}
	\IEEEauthorblockA{\IEEEauthorrefmark{1} Institute of Communications and Navigation, German Aerospace Center, We{\ss}ling, Germany
	}
	\IEEEauthorblockA{\IEEEauthorrefmark{2}Department of Signals and Systems, Chalmers University of Technology, Gothenburg, Sweden
	}
	\IEEEauthorblockA{\IEEEauthorrefmark{3}Department of Electrical and Information Technology, Lund University, Lund, Sweden
	}}

	\maketitle

	\begin{abstract}
		We introduce successive cancellation (SC) decoding of product codes (PCs) with single parity-check (SPC) component codes. Recursive formulas are derived, which resemble the SC decoding algorithm of polar codes. We analyze the error probability of SPC-PCs over the binary erasure channel under SC decoding. A bridge with the analysis of PCs introduced by Elias in 1954 is also established. Furthermore, bounds on the block error probability under SC decoding are provided, and compared to the bounds under the original decoding algorithm proposed by Elias. It is shown that SC decoding of SPC-PCs achieves a lower block error probability than Elias' decoding.
	\end{abstract}

	\section{Introduction}\label{sec:intro}

In 1954, P. Elias showed that the bit error probability over the \ac{BSC} can be made arbitrarily small with a strictly positive coding rate by iterating an infinite number of simple linear block codes, introducing the class of \acp{PC}\cite{Elias:errorfreecoding54}. More recently, \acp{PC}, re-interpreted as turbo-like codes \cite{Pyndiah98}, and their generalizations (see, e.g., \cite{Tanner,Li04,Feltstrom09,Pfister15}) have attracted a large interest from both a research \cite{rankin1,rankin2,1027786,Chiaraluce} and an application \cite{Berrou05} viewpoint.

In \cite{rankin1}, \acp{PC} with \ac{SPC} component codes were decoded using iterative decoding algorithms based on Bahl, Cocke, Jelinek, and Raviv\cite{BCJR} a-posteriori probability decoding of the component codes. In \cite{rankin2}, the asymptotic performance of \ac{SPC}-\acp{PC}, whose component code length doubles with each dimension, was analyzed over the \ac{BSC}, providing an improved bound on the bit error probability by using $2$-dimensional \ac{SPC}-\acp{PC} as the component codes of the overall \ac{PC}. In \cite{6125399}, a bridge between generalized concatenated codes and polar codes is established.

In this paper,  we establish a bridge between the original decoding algorithm of \acp{PC}, which we refer to as Elias' decoder\footnote{By Elias' decoder, we refer to a decoding algorithm that treats the \ac{PC} as a serially concatenated block code, where the decoding is performed starting from the component codes of the first dimension, up to those of the last dimension, in a one-sweep fashion.} \cite{Elias:errorfreecoding54}, and the \ac{SC} decoding algorithm of polar codes \cite{stolte2002rekursive,arikan2009channel}. The link is established for \ac{SPC}-\acp{PC} and for the \ac{BEC}. We show that the block error probability of \ac{SPC}-\acp{PC} can be upper bounded under both decoding algorithms using the evolution of the erasure probabilities over the decoding graph. As a byproduct of the analysis, it is shown that \ac{SPC}-\acp{PC} do not achieve the capacity of the \ac{BEC} under \ac{SC} decoding. A comparison between Elias' decoding and \ac{SC} decoding of \ac{SPC}-\acp{PC} is provided in terms of block error probability. We prove that \ac{SC} decoding yields a lower error probability than Elias' decoding. Finally, simulation results over the \ac{B-AWGN} channel under \ac{SC} decoding are given for different \ac{SPC}-\acp{PC}.
	\section{Preliminaries}\label{sec:prelim}
In the following, $ x_a^b $ denotes the vector $ (x_a, x_{a+1}, \dots, x_b) $ where $ b \geq a $. We use capital letters for \acp{RV} and lower case letters for their realizations. In addition, we denote a \ac{B-DMC} by $ \W : \mathcal{X} \rightarrow \mathcal{Y} $, with input alphabet $ \mathcal{X} = \{0,1\} $, output alphabet $ \mathcal{Y} $, and transition probabilities $ \W(y|x) $, $ x \in \mathcal{X} $, $ y \in \mathcal{Y} $. We write \ac{BEC}($ \epsilon $) to denote the \ac{BEC} with erasure probability $ \epsilon $. The output alphabet of the \ac{BEC} is $\mathcal{Y} = \{0,1,\texttt{e}\}$, where $\texttt{e}$ denotes an erasure.

The generator matrix of an $(n,k)$ \ac{PC} $\code$ is obtained by \emph{iterating} binary linear block codes $\code_1, \code_2, \ldots, \code_m$ in $ m $ dimensions (levels) \cite{Elias:errorfreecoding54}. 
Let $\mathbf{G}_\ell$ be the generator matrix of the $ \ell $-th component code $\code_\ell$. Then, the generator matrix of the $ m $-dimensional \ac{PC} is $ \mathbf{G}=\mathbf{G}_1 \otimes \mathbf{G}_2 \otimes \ldots \otimes \mathbf{G}_m $, where $\otimes$ is the Kronecker product. Upon proper permutation, the generator matrix will permit to encode the message according to the labeling introduced in the next section.

Let $\code_\ell$ be the $\ell$-th component code with parameters $(n_\ell,k_\ell,d_\ell)$, where $n_\ell$, $k_\ell$, and $ d_\ell $ are the block length, dimension, and minimum distance, respectively. Then, the overall \ac{PC} parameters are
\[
n=\prod_{\ell=1}^m n_\ell, \quad k=\prod_{\ell=1}^m k_\ell, \quad \text{and} \quad d=\prod_{\ell=1}^m d_\ell.
\]
Although the characterization of the complete distance spectrum of a \ac{PC} is still an open problem even for the case where the distance spectrum of its component codes is known, the minimum distance multiplicity is known and equal to $ A_{d}=\prod_{\ell=1}^m A^{(\ell)}_{d_\ell} $, where $ A^{(\ell)}_{d_\ell} $ is that of the $ \ell $-th component code. More on the distance spectrum of \acp{PC} can be found in \cite{1027786,Chiaraluce}. Thanks to the relationship between $ d $ and the minimum distances $ d_\ell $ of the component codes, \acp{PC} tend to have a large minimum distance. However, their minimum distance multiplicities are also typically very high \cite{Liva2010:Product_Chinacom}. Note finally that \ac{SPC}-\acp{PC} are a special class of (left-regular) low-density parity-check codes \cite{ldpc}, defined by a bipartite graph with girth $8$.

	\section{Successive Cancellation Decoding of\\Single Parity-Check Product Codes}

Consider transmission over a \ac{B-DMC} $ \W $ using an $(n,k)$ \ac{SPC}-\ac{PC} with $m$ dimensions (levels). Let the binary vectors $ u_1^k $ and $ x_1^n $ be the message to be encoded and the corresponding codeword, respectively, and let $ y_1^n \in \mathcal{Y}^n$ be the channel observation. The transmission by using the $(9,4)$ \ac{SPC}-\ac{PC}, obtained by iterating $(3,2)$ \ac{SPC} codes, is illustrated in Fig. \ref{fig:transmission}. We label the levels by numbers starting from right to left as it is seen for the $2$-dimensional case in the figure. We denote by $\eta_\ell$ the number of local \ac{SPC} codes at level $\ell$, $1\leq\ell\leq m$, which is computed as
\[\eta_\ell = \prod_{i=1}^{\ell-1}(n_i-1)\prod_{i=\ell+1}^{m}n_i.\]
\begin{figure}[]
	\begin{center}
		\begin{tikzpicture}[every node/.style={inner sep=1pt},scale=0.8, every node/.style={scale=0.8}]
\def\dx{0.9}
\def\dy{0.9}
\def\k{5}
\def\coefoff{0.06}
\pgfmathsetmacro\n{pow(2,\k)-1}
\def\radiusempt{1cm}
\def\radiusfull{0.1}

\draw[-] (5*\dx,0*\dy) -- (8.1*\dx,0*\dy); 
\draw[-] (7*\dx,1*\dy) -- (8.1*\dx,1*\dy);
\draw[-] (5*\dx,2*\dy) -- (8.1*\dx,2*\dy);
\draw[-] (7*\dx,0*\dy) -- (7*\dx,2*\dy);

\draw[-] (5*\dx,3*\dy) -- (8.1*\dx,3*\dy);
\draw[-] (7*\dx,4*\dy) -- (8.1*\dx,4*\dy);
\draw[-] (5*\dx,5*\dy) -- (8.1*\dx,5*\dy);
\draw[-] (7*\dx,3*\dy) -- (7*\dx,5*\dy);

\draw[-] (5*\dx,6*\dy) -- (8.1*\dx,6*\dy);
\draw[-] (7*\dx,7*\dy) -- (8.1*\dx,7*\dy);
\draw[-] (5*\dx,8*\dy) -- (8.1*\dx,8*\dy);
\draw[-] (7*\dx,6*\dy) -- (7*\dx,8*\dy);

\node at (-0.5*\dx,7*\dy) {$u_1$};
\node at (-0.5*\dx,5*\dy) {$u_2$};
\node at (-0.5*\dx,3*\dy) {$u_3$};
\node at (-0.5*\dx,1*\dy) {$u_4$};

\node at (8.6*\dx,8*\dy) {$x_1$};
\node at (8.6*\dx,7*\dy) {$x_3$};
\node at (8.6*\dx,6*\dy) {$x_2$};
\node at (8.6*\dx,5*\dy) {$x_7$};
\node at (8.6*\dx,4*\dy) {$x_9$};
\node at (8.6*\dx,3*\dy) {$x_8$};
\node at (8.6*\dx,2*\dy) {$x_4$};
\node at (8.6*\dx,1*\dy) {$x_6$};
\node at (8.6*\dx,0*\dy) {$x_5$};

\node at (11*\dx,8*\dy) {$y_1$};
\node at (11*\dx,7*\dy) {$y_3$};
\node at (11*\dx,6*\dy) {$y_2$};
\node at (11*\dx,5*\dy) {$y_7$};
\node at (11*\dx,4*\dy) {$y_9$};
\node at (11*\dx,3*\dy) {$y_8$};
\node at (11*\dx,2*\dy) {$y_4$};
\node at (11*\dx,1*\dy) {$y_6$};
\node at (11*\dx,0*\dy) {$y_5$};

\draw[-] (9*\dx,0*\dy) -- (10.5*\dx,0*\dy);
\draw[-] (9*\dx,1*\dy) -- (10.5*\dx,1*\dy);
\draw[-] (9*\dx,2*\dy) -- (10.5*\dx,2*\dy);
\draw[-] (9*\dx,3*\dy) -- (10.5*\dx,3*\dy);
\draw[-] (9*\dx,4*\dy) -- (10.5*\dx,4*\dy);
\draw[-] (9*\dx,5*\dy) -- (10.5*\dx,5*\dy);
\draw[-] (9*\dx,6*\dy) -- (10.5*\dx,6*\dy);
\draw[-] (9*\dx,7*\dy) -- (10.5*\dx,7*\dy);
\draw[-] (9*\dx,8*\dy) -- (10.5*\dx,8*\dy);

\node[rectangle,fill=white,draw,minimum size=0.5cm] at (9.75*\dx,0*\dy) {\footnotesize{$W$}};
\node[rectangle,fill=white,draw,minimum size=0.5cm] at (9.75*\dx,1*\dy) {\footnotesize{$W$}};
\node[rectangle,fill=white,draw,minimum size=0.5cm] at (9.75*\dx,2*\dy) {\footnotesize{$W$}};
\node[rectangle,fill=white,draw,minimum size=0.5cm] at (9.75*\dx,3*\dy) {\footnotesize{$W$}};
\node[rectangle,fill=white,draw,minimum size=0.5cm] at (9.75*\dx,4*\dy) {\footnotesize{$W$}};
\node[rectangle,fill=white,draw,minimum size=0.5cm] at (9.75*\dx,5*\dy) {\footnotesize{$W$}};
\node[rectangle,fill=white,draw,minimum size=0.5cm] at (9.75*\dx,6*\dy) {\footnotesize{$W$}};
\node[rectangle,fill=white,draw,minimum size=0.5cm] at (9.75*\dx,7*\dy) {\footnotesize{$W$}};
\node[rectangle,fill=white,draw,minimum size=0.5cm] at (9.75*\dx,8*\dy) {\footnotesize{$W$}};

\node[circle,fill=white,draw,radius=\radiusfull,minimum size=\radiusfull] at (7*\dx,1*\dy) {\tiny{$+$}};
\node[circle,fill=white,draw,radius=\radiusfull,minimum size=\radiusfull] at (7*\dx,4*\dy) {\tiny{$+$}};
\node[circle,fill=white,draw,radius=\radiusfull,minimum size=\radiusfull] at (7*\dx,7*\dy) {\tiny{$+$}};

\draw[-] (0*\dx,1*\dy) -- (3*\dx,1*\dy);
\draw[-] (2*\dx,2*\dy) -- (3*\dx,2*\dy);
\draw[-] (0*\dx,3*\dy) -- (3*\dx,3*\dy);
\draw[-] (2*\dx,1*\dy) -- (2*\dx,3*\dy);

\draw[-] (0*\dx,5*\dy) -- (3*\dx,5*\dy);
\draw[-] (2*\dx,6*\dy) -- (3*\dx,6*\dy);
\draw[-] (0*\dx,7*\dy) -- (3*\dx,7*\dy);
\draw[-] (2*\dx,5*\dy) -- (2*\dx,7*\dy);

\node[circle,fill=white,draw,radius=\radiusfull,minimum size=\radiusfull] at (2*\dx,2*\dy) {\tiny{$+$}};
\node[circle,fill=white,draw,radius=\radiusfull,minimum size=\radiusfull] at (2*\dx,6*\dy) {\tiny{$+$}};

\draw[-] (3*\dx,7*\dy) -- (5*\dx,8*\dy);
\draw[-] (3*\dx,6*\dy) -- (5*\dx,5*\dy);
\draw[-] (3*\dx,5*\dy) -- (5*\dx,2*\dy);

\draw[-] (3*\dx,3*\dy) -- (5*\dx,6*\dy);
\draw[-] (3*\dx,2*\dy) -- (5*\dx,3*\dy);
\draw[-] (3*\dx,1*\dy) -- (5*\dx,0*\dy);

\draw[rounded corners=1mm,dashed,gray,very thin] (6.0*\dx,-0.35*\dy) rectangle (7.5*\dx,2.35*\dy);
\draw[rounded corners=1mm,dashed,gray,very thin] (6.0*\dx,2.65*\dy) rectangle (7.5*\dx,5.35*\dy);
\draw[rounded corners=1mm,dashed,gray,very thin] (6.0*\dx,5.65*\dy) rectangle (7.5*\dx,8.35*\dy);

\draw[rounded corners=1mm,dashed,gray,very thin] (1.0*\dx,0.65*\dy) rectangle (2.5*\dx,3.35*\dy);
\draw[rounded corners=1mm,dashed,gray,very thin] (1.0*\dx,4.65*\dy) rectangle (2.5*\dx,7.35*\dy);

\draw[rounded corners=1mm,dashed,gray,very thin] (0.25*\dx,-0.5*\dy) rectangle (7.9*\dx,8.5*\dy);

\node at (3.95*\dx,9.15*\dy) {\footnotesize{$(9,4)$ SPC product code encoder}};
\draw[->] (3.95*\dx,9.0*\dy) -- (3.95*\dx,8.52*\dy);

\node at (2.25*\dx,8*\dy) {\scriptsize{$(3,2)$ SPC local code encoder}};
\draw[->] (1.85*\dx,7.85*\dy) -- (1.85*\dx,7.37*\dy);

\node at (2*\dx,-1*\dy) {$2$nd level};
\node at (7*\dx,-1*\dy) {$1$st level};

\end{tikzpicture}
	\end{center}
	\caption{Transmission by using the $(9,4)$ \ac{SPC}-\ac{PC}.}
	\label{fig:transmission}
\end{figure}
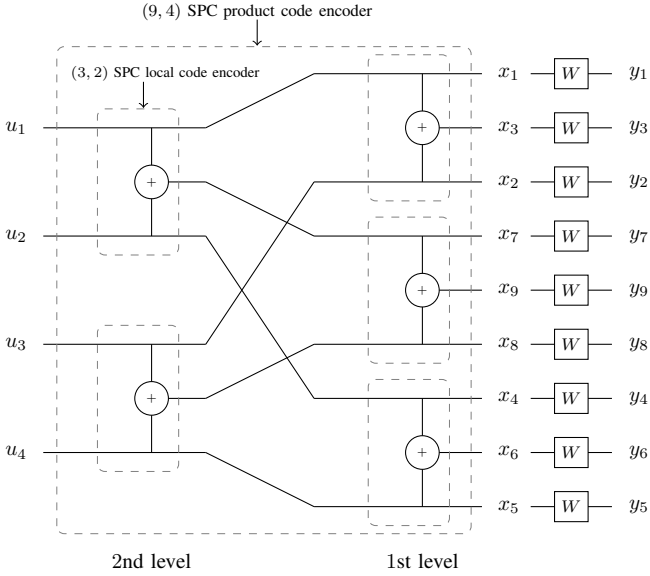

\ac{SC} decoding follows the schedule introduced in \cite{arikan2009channel} for polar codes. Explicitly, the decision on the $i$-th information bit is made according to the soft-information obtained by performing a message-passing algorithm which propagates messages from the right of the decoding graph, along the edges of the tree rooted in $u_i$, where the operations at the local codes take into account the past decisions $\hat{u}_1^{i-1}$. The decoding graph for the $(9,4)$ \ac{SPC}-\ac{PC} is provided in Fig. \ref{fig:decoding_graph}, by introducing \acp{CN} and \acp{VN} to the encoder graph given in Fig. \ref{fig:transmission}. We denote the soft-messages coming from the right, associated with the $i$-th codeword bit of the $j$-th local code at level $\ell$, as $\rho_{j,i}^{(\ell)}$, $1\leq i \leq n_\ell$, $1\leq j\leq \eta_\ell$. The inputs to the decoder are defined as the channel \acp{LLR}, i.e.,
\[\rho_{j,i}^{(1)} \triangleq \ln\bigg[\frac{\W(y_{(j-1)n_1+i}|0)}{\W(y_{(j-1)n_1+i}|1)}\bigg]\]
$1\leq j\leq \eta_1$ and $1\leq i\leq n_1$.

\begin{figure}[]
	\begin{center}
		\begin{tikzpicture}[every node/.style={inner sep=1pt},scale=0.8, every node/.style={scale=0.8}]
\def\dx{0.9}
\def\dy{0.9}
\def\k{5}
\def\coefoff{0.06}
\pgfmathsetmacro\n{pow(2,\k)-1}
\def\radiusempt{1cm}
\def\radiusfull{0.5cm}

\draw[-] (7*\dx,0*\dy) -- (7*\dx,2*\dy);
\draw[-] (7*\dx,3*\dy) -- (7*\dx,5*\dy);
\draw[-] (7*\dx,6*\dy) -- (7*\dx,8*\dy);

\draw[-] (2*\dx,5*\dy) -- (2*\dx,7*\dy);
\draw[-] (2*\dx,1*\dy) -- (2*\dx,3*\dy);

\node[circle,fill=white,draw,radius=\radiusfull,minimum size=\radiusfull] (16=) at (7*\dx,0*\dy) {\footnotesize{$=$}};
\node[circle,fill=white,draw,radius=\radiusfull,minimum size=\radiusfull] (15=) at (7*\dx,2*\dy) {\footnotesize{$=$}};
\node[circle,fill=white,draw,radius=\radiusfull,minimum size=\radiusfull] (14=) at (7*\dx,3*\dy) {\footnotesize{$=$}};
\node[circle,fill=white,draw,radius=\radiusfull,minimum size=\radiusfull] (13=) at (7*\dx,5*\dy) {\footnotesize{$=$}};
\node[circle,fill=white,draw,radius=\radiusfull,minimum size=\radiusfull] (12=) at (7*\dx,6*\dy) {\footnotesize{$=$}};
\node[circle,fill=white,draw,radius=\radiusfull,minimum size=\radiusfull] (11=) at (7*\dx,8*\dy) {\footnotesize{$=$}};

\node[circle,fill=white,draw,radius=\radiusfull,minimum size=\radiusfull] (24=) at (2*\dx,1*\dy) {\footnotesize{$=$}};
\node[circle,fill=white,draw,radius=\radiusfull,minimum size=\radiusfull] (23=) at (2*\dx,3*\dy) {\footnotesize{$=$}};
\node[circle,fill=white,draw,radius=\radiusfull,minimum size=\radiusfull] (22=) at (2*\dx,5*\dy) {\footnotesize{$=$}};
\node[circle,fill=white,draw,radius=\radiusfull,minimum size=\radiusfull] (21=) at (2*\dx,7*\dy) {\footnotesize{$=$}};

\node[circle,fill=white,draw,radius=\radiusfull,minimum size=\radiusfull] (13+) at (7*\dx,1*\dy) {\tiny{$+$}};
\node[circle,fill=white,draw,radius=\radiusfull,minimum size=\radiusfull] (12+) at (7*\dx,4*\dy) {\tiny{$+$}};
\node[circle,fill=white,draw,radius=\radiusfull,minimum size=\radiusfull] (11+) at (7*\dx,7*\dy) {\tiny{$+$}};
\node[circle,fill=white,draw,radius=\radiusfull,minimum size=\radiusfull] (22+) at (2*\dx,2*\dy) {\tiny{$+$}};
\node[circle,fill=white,draw,radius=\radiusfull,minimum size=\radiusfull] (21+) at (2*\dx,6*\dy) {\tiny{$+$}};

\draw[-] (5*\dx,0*\dy) -- (16=);
\draw[-latex] (8*\dx,0*\dy) -- (16=);

\draw[-latex] (8*\dx,1*\dy) -- (13+);
\draw[-] (5*\dx,2*\dy) -- (15=);
\draw[-latex] (8*\dx,2*\dy) -- (15=);

\draw[-] (5*\dx,3*\dy) -- (14=);
\draw[-latex] (8*\dx,3*\dy) -- (14=);

\draw[-latex] (8*\dx,4*\dy) -- (12+);
\draw[-] (5*\dx,5*\dy) -- (13=);
\draw[-latex] (8*\dx,5*\dy) -- (13=);

\draw[-] (5*\dx,6*\dy) -- (12=);
\draw[-latex] (8*\dx,6*\dy) -- (12=);

\draw[-latex] (8*\dx,7*\dy) -- (11+);
\draw[-] (5*\dx,8*\dy) -- (11=);
\draw[-latex] (8*\dx,8*\dy) -- (11=);

\node at (10*\dx,8*\dy) {$\rho_{1,1}^{(1)} \triangleq \ln\Big[\frac{\W(y_1|x_1=0)}{\W(y_1|x_1=1)}\Big]$};
\node at (10*\dx,7*\dy) {$\rho_{1,3}^{(1)} \triangleq \ln\Big[\frac{\W(y_3|x_3=0)}{\W(y_3|x_3=1)}\Big]$};
\node at (10*\dx,6*\dy) {$\rho_{1,2}^{(1)} \triangleq \ln\Big[\frac{\W(y_2|x_2=0)}{\W(y_2|x_2=1)}\Big]$};
\node at (10*\dx,5*\dy) {$\rho_{3,1}^{(1)} \triangleq \ln\Big[\frac{\W(y_7|x_7=0)}{\W(y_7|x_7=1)}\Big]$};
\node at (10*\dx,4*\dy) {$\rho_{3,3}^{(1)} \triangleq \ln\Big[\frac{\W(y_9|x_9=0)}{\W(y_9|x_9=1)}\Big]$};
\node at (10*\dx,3*\dy) {$\rho_{3,2}^{(1)} \triangleq \ln\Big[\frac{\W(y_8|x_8=0)}{\W(y_8|x_8=1)}\Big]$};
\node at (10*\dx,2*\dy) {$\rho_{2,1}^{(1)} \triangleq \ln\Big[\frac{\W(y_4|x_4=0)}{\W(y_4|x_4=1)}\Big]$};
\node at (10*\dx,1*\dy) {$\rho_{2,3}^{(1)} \triangleq \ln\Big[\frac{\W(y_6|x_6=0)}{\W(y_6|x_6=1)}\Big]$};
\node at (10*\dx,0*\dy) {$\rho_{2,2}^{(1)} \triangleq \ln\Big[\frac{\W(y_5|x_5=0)}{\W(y_5|x_5=1)}\Big]$};

\node at (7*\dx,-1*\dy) {$1$st level};

\node at (2*\dx,-1*\dy) {$2$nd level};

\draw[-] (1*\dx,1*\dy) -- (24=);
\draw[-] (3*\dx,1*\dy) -- (24=);
\draw[-] (3*\dx,2*\dy) -- (22+);
\draw[-] (1*\dx,3*\dy) -- (23=);
\draw[-] (3*\dx,3*\dy) -- (23=);

\draw[-] (1*\dx,5*\dy) -- (22=);
\draw[-] (3*\dx,5*\dy) -- (22=);
\draw[-] (3*\dx,6*\dy) -- (21+);
\draw[-] (1*\dx,7*\dy) -- (21=);
\draw[-] (3*\dx,7*\dy) -- (21=);

\draw[-] (3*\dx,7*\dy) -- (5*\dx,8*\dy);
\draw[-] (3*\dx,6*\dy) -- (5*\dx,5*\dy);
\draw[-] (3*\dx,5*\dy) -- (5*\dx,2*\dy);

\draw[-] (3*\dx,3*\dy) -- (5*\dx,6*\dy);
\draw[-] (3*\dx,2*\dy) -- (5*\dx,3*\dy);
\draw[-] (3*\dx,1*\dy) -- (5*\dx,0*\dy);

\end{tikzpicture}
	\end{center}
	\caption{Decoding graph for the $(9,4)$ \ac{SPC}-\ac{PC}. \mytikzVN\, denotes a \ac{VN} and \mytikzCN\, denotes a \ac{CN}.}
	\label{fig:decoding_graph}
\end{figure}
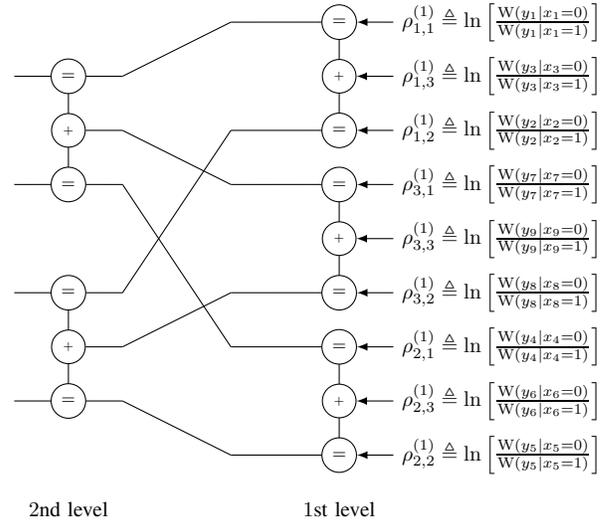

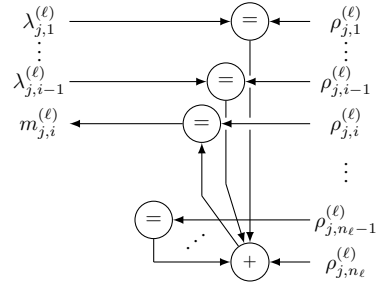
\begin{figure}
	\centering
		\begin{tikzpicture}[every node/.style={inner sep=1pt},scale=0.8, every node/.style={scale=0.8}]
	
	\node[circle,fill=white,draw,radius=0.5,minimum size=0.5] (eq1) at (2,4) {$=$};
	\node[circle,fill=white,draw,radius=0.5,minimum size=0.5] (eq2) at (1.6,3) {$=$};
	\node[circle,fill=white,draw,radius=0.5,minimum size=0.5] (eq3) at (1.2,2.3) {$=$};
	\node[circle,fill=white,draw,radius=0.5,minimum size=0.5] (eq4) at (0.4,0.7) {$=$};
	\node[circle,fill=white,draw,radius=0.5,minimum size=0.5] (sum) at (2,0) {\footnotesize{$+$}};
	\draw[-latex] (3,4) -- (eq1);
	\draw[-latex] (3,3) -- (eq2);
	\draw[-latex] (3,2.3) -- (eq3);
	\draw[-latex] (3,0.7) -- (eq4);			
	\draw[-latex] (3,0) -- (sum);
	
	\draw[-] (eq1) -- (2,3.1);\draw[-] (2,2.9) -- (2,2.4);\draw[-latex] (2,2.2) -- (sum);
	\draw[-] (eq2) -- (1.6,2.4);\draw[-] (1.6,2.2) -- (1.6,1.3);\draw[-latex] (1.6,1.3) -- (sum);
	\draw[-latex] (1.2,1.1) -- (eq3);\draw[-] (1.2,1.1) -- (sum);
	
	\draw[-] (eq4) -- (0.4,0);\draw[-latex] (0.4,0) -- (sum);
	
	\draw[-latex] (-1,4) -- (eq1);
	\draw[-latex] (-1,3) -- (eq2);
	\draw[-latex] (eq3) -- (-1,2.3);
	
	\node at (-1.5,4) {$\lambda_{j,1}^{(\ell)}$};
	\node at (-1.5,3.6) {\vdots};
	\node at (-1.5,3) {$\lambda_{j,i-1}^{(\ell)}$};
	\node at (-1.5,2.3) {$m_{j,i}^{(\ell)}$};
	\node at (3.6,4) {$\rho_{j,1}^{(\ell)}$};
	\node at (3.6,3.6) {\vdots};
	\node at (3.6,3) {$\rho_{j,i-1}^{(\ell)}$};
	\node at (3.6,2.3) {$\rho_{j,i}^{(\ell)}$};
	\node at (3.6,1.6) {\vdots};
	\node at (3.6,0.7) {$\rho_{j,n_\ell-1}^{(\ell)}$};
	\node at (3.6,0) {$\rho_{j,n_\ell}^{(\ell)}$};
	
	\node at (1.1,0.45) {\reflectbox{$\ddots$}};
	
	\end{tikzpicture}
	\caption{The decoding graph of the $j$-th local code at level $\ell$.}
	\label{fig:local_graph}
\end{figure}

Consider the $j$-th local code at level $\ell$, whose decoding graph is provided in Fig. \ref{fig:local_graph}. The soft output message for the $i$-th local information bit is computed as
\begin{equation}
	m_{j,i}^{(\ell)} = \rho_{j,i}^{(\ell)} + 2\atanh\Bigg[\prod_{i^\prime=i+1}^{n_\ell}\tanh\bigg(\frac{\rho_{j,i^\prime}^{(\ell)}}{2}\bigg)\Bigg]\Big(-1\Big)^{\sum_{z=1}^{i-1}\lambda_{j,z}^{(\ell)}} \label{eq:metric_calc}
\end{equation}
where $\lambda_{j,z}^{(\ell)}$ is the hard input (i.e., bit) message for the $z$-th local information bit, coming from the left, with $z = 1,\dots,i-1$, depending on the past decisions. The computed output message $m_{j,i}^{(\ell)}$ is propagated leftwards over the tree edge, providing the next level with an input message. In particular, we set
\begin{equation}
	\rho_{j^\prime,i^\prime}^{(\ell+1)} \leftarrow m_{j,i}^{(\ell)}
	\label{eq:permutation}
\end{equation}
where the assignment is made according to the graph connections, i.e., the \ac{PC} structure. The decision is made as

\begin{equation}
\hat{u}_{(j-1)(n_m-1)+i} = \left\{\begin{array}{lll}
0 &\text{if}& m_{j,i}^{(m)} \geq 0 \\
1 &\text{if}& m_{j,i}^{(m)} < 0
\end{array}
\right.
\label{eq:decision}
\end{equation}
by breaking the ties in favor of zero. Over the \ac{BEC}, ties are not broken towards any decision by revising \eqref{eq:decision} as
\begin{equation}
\hat{u}_{(j-1)(n_m-1)+i} = \left\{\begin{array}{lll}
0 &\text{if}& m_{j,i}^{(m)} = +\infty \\
\texttt{e} & \text{if}& m_{j,i}^{(m)} = 0 \\
1 &\text{if}& m_{j,i}^{(m)} = -\infty.
\end{array}
\right.
\label{eq:decision_BEC}
\end{equation}
Accordingly, over the \ac{BEC}, \eqref{eq:metric_calc} is valid if $\lambda_{j,z}^{(\ell)}\neq\texttt{e}$ for all $z = 1,\dots,i-1$. However, if there exists any $z = 1,\dots,i-1$, such that $\lambda_{j,z}^{(\ell)} = \texttt{e}$, then \eqref{eq:metric_calc} has to be replaced by
\begin{equation}
	m_{j,i}^{(\ell)} = \rho_{j,i}^{(\ell)}.
	\label{eq:metric_calc_BEC}
\end{equation}
A block error event occurs if $ \hat{u}_1^k \neq u_1^k $.
\begin{example}
	Consider the $2$-dimensional $(9,4)$ \ac{SPC}-\ac{PC} obtained by iterating $(3,2)$ \ac{SPC} codes. Its decoding graph is provided in Fig. \ref{fig:decoding_graph}. The number of local codes at levels $1$ and $2$ are computed respectively as $ \eta_1 = 3 $ and $\eta_2 = 2$. We illustrate an intermediate \ac{SC} decoding step in Fig. \ref{fig:intermediate}, where the decisions for the first two information bits are already made and the decoder computes the soft message for the third bit. At level $1$, the $j$-th local code has the hard input message $\lambda_{j,1}^{(1)}$ and the soft input messages $\left(\rho_{j,1}^{(1)},\rho_{j,2}^{(1)},\rho_{j,3}^{(1)}\right)$. Using \eqref{eq:metric_calc} or \eqref{eq:metric_calc_BEC} depending on the previous decisions, it computes the soft output message $m_{j,2}^{(1)}$, providing the next level with a soft input message. According to the connections in the graph, we have the following assignments:
	\[\rho_{2,1}^{(2)}\leftarrow m_{1,2}^{(1)},\quad\rho_{2,2}^{(2)}\leftarrow m_{2,2}^{(1)},\quad\text{and}\quad\rho_{2,3}^{(2)}\leftarrow m_{3,2}^{(1)}.\]
	Then, the soft output message $m_{2,1}^{(2)}$ is computed with the soft input messages coming from the right. The final decision is made for $\hat{u}_3$ as in \eqref{eq:decision} or \eqref{eq:decision_BEC} depending on the channel.
	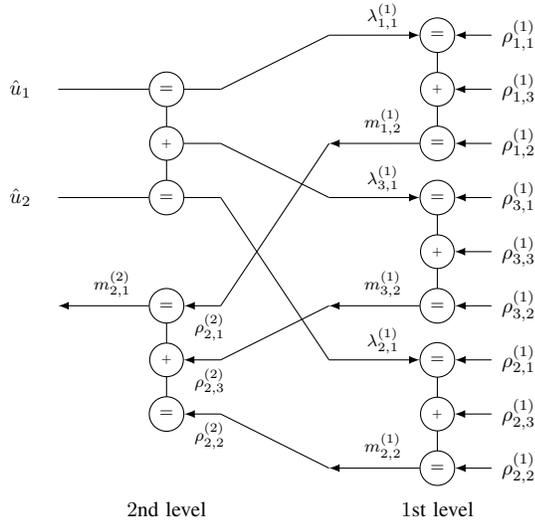
\begin{figure}
		\centering
			\begin{tikzpicture}[every node/.style={inner sep=1pt},scale=0.8, every node/.style={scale=0.8}]
\def\dx{0.9}
\def\dy{0.9}
\def\k{5}
\def\coefoff{0.06}
\pgfmathsetmacro\n{pow(2,\k)-1}
\def\radiusempt{1cm}
\def\radiusfull{0.5cm}

\draw[-] (7*\dx,0*\dy) -- (7*\dx,2*\dy);
\draw[-] (7*\dx,3*\dy) -- (7*\dx,5*\dy);
\draw[-] (7*\dx,6*\dy) -- (7*\dx,8*\dy);

\draw[-] (2*\dx,5*\dy) -- (2*\dx,7*\dy);
\draw[-] (2*\dx,1*\dy) -- (2*\dx,3*\dy);

\node[circle,fill=white,draw,radius=\radiusfull,minimum size=\radiusfull] (16=) at (7*\dx,0*\dy) {\footnotesize{$=$}};
\node[circle,fill=white,draw,radius=\radiusfull,minimum size=\radiusfull] (15=) at (7*\dx,2*\dy) {\footnotesize{$=$}};
\node[circle,fill=white,draw,radius=\radiusfull,minimum size=\radiusfull] (14=) at (7*\dx,3*\dy) {\footnotesize{$=$}};
\node[circle,fill=white,draw,radius=\radiusfull,minimum size=\radiusfull] (13=) at (7*\dx,5*\dy) {\footnotesize{$=$}};
\node[circle,fill=white,draw,radius=\radiusfull,minimum size=\radiusfull] (12=) at (7*\dx,6*\dy) {\footnotesize{$=$}};
\node[circle,fill=white,draw,radius=\radiusfull,minimum size=\radiusfull] (11=) at (7*\dx,8*\dy) {\footnotesize{$=$}};

\node[circle,fill=white,draw,radius=\radiusfull,minimum size=\radiusfull] (24=) at (2*\dx,1*\dy) {\footnotesize{$=$}};
\node[circle,fill=white,draw,radius=\radiusfull,minimum size=\radiusfull] (23=) at (2*\dx,3*\dy) {\footnotesize{$=$}};
\node[circle,fill=white,draw,radius=\radiusfull,minimum size=\radiusfull] (22=) at (2*\dx,5*\dy) {\footnotesize{$=$}};
\node[circle,fill=white,draw,radius=\radiusfull,minimum size=\radiusfull] (21=) at (2*\dx,7*\dy) {\footnotesize{$=$}};

\node[circle,fill=white,draw,radius=\radiusfull,minimum size=\radiusfull] (13+) at (7*\dx,1*\dy) {\tiny{$+$}};
\node[circle,fill=white,draw,radius=\radiusfull,minimum size=\radiusfull] (12+) at (7*\dx,4*\dy) {\tiny{$+$}};
\node[circle,fill=white,draw,radius=\radiusfull,minimum size=\radiusfull] (11+) at (7*\dx,7*\dy) {\tiny{$+$}};
\node[circle,fill=white,draw,radius=\radiusfull,minimum size=\radiusfull] (22+) at (2*\dx,2*\dy) {\tiny{$+$}};
\node[circle,fill=white,draw,radius=\radiusfull,minimum size=\radiusfull] (21+) at (2*\dx,6*\dy) {\tiny{$+$}};

\draw[-latex] (16=) -- (5*\dx,0*\dy);
\draw[-latex] (8*\dx,0*\dy) -- (16=);

\draw[-latex] (8*\dx,1*\dy) -- (13+);
\draw[-latex] (5*\dx,2*\dy) -- (15=);
\draw[-latex] (8*\dx,2*\dy) -- (15=);

\draw[-latex] (14=) -- (5*\dx,3*\dy);
\draw[-latex] (8*\dx,3*\dy) -- (14=);

\draw[-latex] (8*\dx,4*\dy) -- (12+);
\draw[-latex] (5*\dx,5*\dy) -- (13=);
\draw[-latex] (8*\dx,5*\dy) -- (13=);

\draw[-latex] (12=) -- (5*\dx,6*\dy);
\draw[-latex] (8*\dx,6*\dy) -- (12=);

\draw[-latex] (8*\dx,7*\dy) -- (11+);
\draw[-latex] (5*\dx,8*\dy) -- (11=);
\draw[-latex] (8*\dx,8*\dy) -- (11=);

\node at (-0.7*\dx,7*\dy) {$\hat{u}_1$};
\node at (-0.7*\dx,5*\dy) {$\hat{u}_2$};

\node at (7*\dx,-0.75*\dy) {$1$st level};

\node at (2*\dx,-0.75*\dy) {$2$nd level};

\node at (8.5*\dx,8*\dy) {$\rho_{1,1}^{(1)}$};
\node at (8.5*\dx,7*\dy) {$\rho_{1,3}^{(1)}$};
\node at (8.5*\dx,6*\dy) {$\rho_{1,2}^{(1)}$};
\node at (8.5*\dx,5*\dy) {$\rho_{3,1}^{(1)}$};
\node at (8.5*\dx,4*\dy) {$\rho_{3,3}^{(1)}$};
\node at (8.5*\dx,3*\dy) {$\rho_{3,2}^{(1)}$};
\node at (8.5*\dx,2*\dy) {$\rho_{2,1}^{(1)}$};
\node at (8.5*\dx,1*\dy) {$\rho_{2,3}^{(1)}$};
\node at (8.5*\dx,0*\dy) {$\rho_{2,2}^{(1)}$};

\node at (6*\dx,6.4*\dy) {\footnotesize{$m_{1,2}^{(1)}$}};
\node at (6*\dx,3.4*\dy) {\footnotesize{$m_{3,2}^{(1)}$}};
\node at (6*\dx,0.4*\dy){\footnotesize{$m_{2,2}^{(1)}$}};

\node at (1*\dx,3.4*\dy) {\footnotesize{$m_{2,1}^{(2)}$}};

\node at (2.8*\dx,2.6*\dy) {\footnotesize{$\rho_{2,1}^{(2)}$}};
\node at (2.8*\dx,1.6*\dy) {\footnotesize{$\rho_{2,3}^{(2)}$}};
\node at (2.8*\dx,0.6*\dy) {\footnotesize{$\rho_{2,2}^{(2)}$}};

\node at (6*\dx,8.35*\dy) {\footnotesize{$\lambda_{1,1}^{(1)}$}};
\node at (6*\dx,5.35*\dy) {\footnotesize{$\lambda_{3,1}^{(1)}$}};
\node at (6*\dx,2.35*\dy) {\footnotesize{$\lambda_{2,1}^{(1)}$}};

\draw[-latex] (3*\dx,1*\dy) -- (24=);
\draw[-latex] (3*\dx,2*\dy) -- (22+);
\draw[-latex] (23=) -- (0*\dx,3*\dy);
\draw[-latex] (3*\dx,3*\dy) -- (23=);

\draw[-] (0*\dx,5*\dy) -- (22=);
\draw[-] (3*\dx,5*\dy) -- (22=);
\draw[-] (3*\dx,6*\dy) -- (21+);
\draw[-] (0*\dx,7*\dy) -- (21=);
\draw[-] (3*\dx,7*\dy) -- (21=);

\draw[-] (3*\dx,7*\dy) -- (5*\dx,8*\dy);
\draw[-] (3*\dx,6*\dy) -- (5*\dx,5*\dy);
\draw[-] (3*\dx,5*\dy) -- (5*\dx,2*\dy);

\draw[-] (3*\dx,3*\dy) -- (5*\dx,6*\dy);
\draw[-] (3*\dx,2*\dy) -- (5*\dx,3*\dy);
\draw[-] (3*\dx,1*\dy) -- (5*\dx,0*\dy);

\end{tikzpicture}
		\caption{Decoding graph for the third information bit $u_3$.}
		\label{fig:intermediate}
	\end{figure}
\end{example}
We revise \eqref{eq:metric_calc} under Elias' decoder as
\begin{equation}
m_{j,i}^{(\ell)} = \rho_{j,i}^{(\ell)} + 2\atanh\Bigg[\prod_{i^\prime\neq i,i^\prime = 1}^{n_\ell}\tanh\bigg(\frac{\rho_{j,i^\prime}^{(j)}}{2}\bigg)\Bigg].
\label{eq:metric_calc_ED}
\end{equation}
\subsection{Analysis over the Binary Erasure Channel}

We first analyze the behavior of a local $(n_\ell,n_\ell-1)$ \ac{SPC} (component) code under \ac{SC} decoding over the \ac{BEC}($\epsilon$) $ \W $. We denote by $\epsilon^{(i)}$ the erasure probability of the $i$-th information bit after \ac{SC} decoding conditioned on the correct decoding of the $i-1$ preceding bits, with $i=1,\ldots , n_\ell-1$. The relationship between the input-output erasure probabilities is given by
\begin{equation}
	\epsilon^{(i)}=\epsilon\left(1-\left(1-\epsilon\right)^{n_\ell-i}\right)\qquad \text{for}\quad i=1,\dots,n_\ell-1 
	\label{eq:erasure_SPC}.
\end{equation}
Denote the bits at the input of such a local SPC code encoder as $v_1,\dots,v_{n_\ell-1}$ and the received values at the input of the corresponding local SPC decoder as $b_1,\dots,b_{n_\ell}$. Let $\mi^{(i)}$ denote the mutual information between the \acp{RV} $V_i$ and $(B_1^{n_\ell},V_1^{i-1})$, i.e., $\mi^{(i)} \triangleq \mi(V_i;B_1^{n_\ell},V_1^{i-1})$. Then, the recursion \eqref{eq:erasure_SPC} can be rewritten in terms of mutual information as
\begin{equation}
	\mi^{(i)}=1-\left(1-\mi\right)\left(1-\mi^{n_\ell-i}\right)\qquad \text{for}\quad i=1,\ldots, n_\ell-1 
	\label{eq:mi_SPC}
\end{equation}
with $\mi\triangleq 1-\epsilon$  and $\mi^{(i)}\triangleq 1-\epsilon^{(i)}$.
\begin{prop}\label{prop:mi_loss}
	The mutual information at the input of an \ac{SPC} local decoder in the $\ell$-th dimension is not preserved at its output, i.e.,
	\[
		\sum_{j=1}^{n_\ell-1}\mi^{(j)}<n_\ell \mi.
	\] 
\end{prop}
\begin{proof}
	We have that
	\begin{align*}
		\sum_{j=1}^{n_\ell-1}\mi^{(j)} &= \sum_{j=1}^{n_\ell-1} \left[1-\left(1-\mi\right)\left(1-\mi^{n_\ell-j}\right)\right] \\
		&=(n_\ell-1)-\left(1-\mi\right)\sum_{j=1}^{n_\ell-1}\left(1-\mi^{n_\ell-j}\right)\\
		&=(n_\ell-1)\mi + \left(1-\mi\right)\mi \frac{1-\mi^{n_\ell-1}}{1-\mi}\\[3\jot]
		&=n_\ell\mi  - \mi^{n_\ell} < n_\ell\mi. &&\qedhere
	\end{align*}
\end{proof}
Proposition \ref{prop:mi_loss} provides also the loss of mutual information due to a local \ac{SPC} code in the $\ell$-th dimension, which is $\mi^{n_\ell}$. By recursively applying the transformation \eqref{eq:erasure_SPC}, one can derive the erasure probability associated with the information bits of an \ac{SPC}-\ac{PC}. Denote such erasure probabilities as $q_i$, $i=1,\ldots,k$. The evolution of the corresponding mutual information values at each transformation level is illustrated in Fig. \ref{fig:PEVO}, where $ \mi = 0.3 $, i.e., $ \epsilon = 0.7 $, for the original channel.

\begin{figure}[]
	\begin{center}
		\newcommand{\length}{9}
\begin{tikzpicture}[every node/.style={inner sep=1pt},scale=0.9, every node/.style={transform shape},font=\footnotesize]
\def\dx{0.68}
\def\dy{6}
\draw[-latex,black] (1*\dx,0*\dy) -- (1*\dx,1.2*\dy);
\draw[-latex,black] (1*\dx,0*\dy) -- (12*\dx,0*\dy);
\draw[gray,dashed] (1*\dx,1*\dy) -- (11*\dx,1*\dy);
\node[rotate=90] at (-0.3*\dx,0.5*\dy) {{Mutual Information}};
\node at (0.7*\dx,1*\dy) {{$1$}};
\node at (0.4*\dx,0.3*\dy) {{$0.3$}};
\node at (1*\dx,-0.06*\dy) {{$0$}};
\node at (6*\dx,-0.11*\dy) {{level}};

\foreach \n in {1,...,\length}{
\pgfplotstableread{level\n.txt}{\tableA}
\pgfplotstablegetrowsof{level\n.txt}
\pgfmathsetmacro{\rowsA}{\pgfplotsretval-1}

\node at (\n*\dx+\dx,-0.06*\dy) {\footnotesize{$\n$}};
\draw[black]  (\n*\dx+\dx,0*\dy) -- (\n*\dx+\dx,0.02*\dy);

\foreach \i in {0,...,\rowsA}{
	\pgfplotstablegetelem{\i}{[index] 0}\of{\tableA}
		\let\xs\pgfplotsretval
	\pgfplotstablegetelem{\i}{[index] 1}\of{\tableA}
		\let\Gs\pgfplotsretval
	\pgfplotstablegetelem{\i}{[index] 2}\of{\tableA}
		\let\Bs\pgfplotsretval
	\pgfplotstablegetelem{\i}{[index] 3}\of{\tableA}
		\let\xe\pgfplotsretval
	\pgfplotstablegetelem{\i}{[index] 4}\of{\tableA}
		\let\Ge\pgfplotsretval
	\pgfplotstablegetelem{\i}{[index] 5}\of{\tableA}
		\let\Be\pgfplotsretval
	\draw[ultra thin,black!90,dashed] (\dx*\xs,\dy*\Bs) -- (\dx*\xe,\dy*\Be);
	\draw[ultra thin,black!90] (\dx*\xs,\dy*\Gs) -- (\dx*\xe,\dy*\Ge);	
}

}\end{tikzpicture}
	\end{center}
	\caption{Mutual information evolution via $(3,2)$ \ac{SPC} codes. At each level, two mutual information values are computed from a root by using \eqref{eq:mi_SPC}. The dashed line segment shows the evolution of the mutual information corresponding to the higher erasure probability while the solid one corresponding to the lower.}\label{fig:PEVO}
\end{figure}
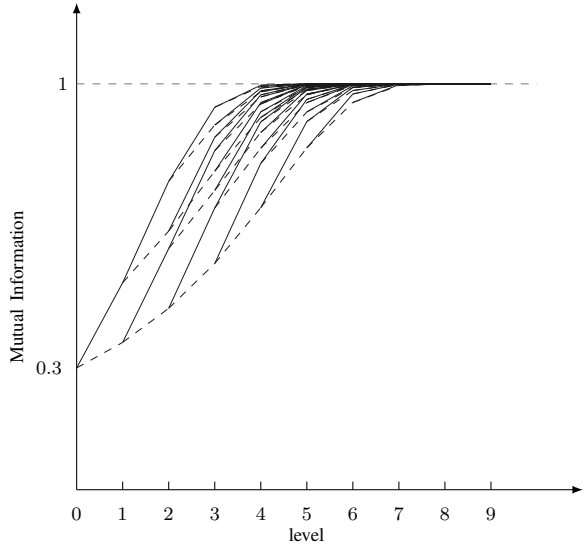

The largest bit erasure probability is equal to that of the first decoded information bit, i.e., $ q_{\mathrm{max}}\triangleq \max_{i=1,\ldots,k} q_i = q_1 $. The block error probability under \ac{SC} decoding, denoted by $P_{\SC}$, is bounded as \cite{arikan2009channel}
\begin{equation}
	q_{\mathrm{max}}\leq P_{\SC} \leq \sum_{i=1}^k q_i \label{eq:UB_Product}.
\end{equation}
A looser upper bound can be obtained by tracking only the largest erasure probability as
\begin{equation}
	P_{\SC} \leq kq_{\mathrm{max}}. \label{eq:LooseUB_Product}
\end{equation}
Note that the derivation of $q_{\mathrm{max}}$ is obtained by iterating the transformation for the first decoded information bit only, i.e.,
\begin{equation}
	\epsilon_\ell^{(1)}=\epsilon_{\ell-1}^{(1)}\left(1-\left(1-\epsilon_{\ell-1}^{(1)}\right)^{n_\ell-1}\right) \quad \text{for}\quad \ell\geq 1 \label{eq:Worst_Channel_Arikan}
\end{equation}
where $\ell$ is the transformation level, and $ \epsilon_0^{(1)} = \epsilon $. For an $ m $-dimensional \ac{SPC}-\ac{PC}, the erasure probability of the first decoded information bit is
\begin{equation}
	q_{\mathrm{max}}= \epsilon_m^{(1)}. \label{eq:Worst_Channel_Arikan_at_last_stage}
\end{equation}
Remarkably, \eqref{eq:Worst_Channel_Arikan} describes also the evolution of the bit erasure probability under bounded distance decoding at each level, according to the Elias' decoder analysis\cite{Elias:errorfreecoding54}. 

\begin{lemma}
	For an \ac{SPC}-\ac{PC}, the erasure probability of the first decoded information bit under \ac{SC} decoding \big(given by \eqref{eq:Worst_Channel_Arikan_at_last_stage}\big) is equal to the erasure probability of each information bit under Elias' decoding. \label{lem:worst_channel}
\end{lemma}

We skip the proof as it is intuitive. \ac{SC} decoding makes use of the observation $ y_1^n $ only to decode the first information bit as it is the case for Elias' decoding to decode each information bit \big(see \eqref{eq:metric_calc_ED}\big). However, \ac{SC} decoding exploits also the decisions on the preceding bits to decode the other information bits \big(see \eqref{eq:metric_calc_BEC}\big). As a result of Lemma \ref{lem:worst_channel}, the bound \eqref{eq:LooseUB_Product} holds also for Elias' decoding.

\begin{theorem}
	The block error probability $P_{\E}$ of an \ac{SPC}-\ac{PC} over the \ac{BEC}($ \epsilon $) under Elias' decoding \cite{Elias:errorfreecoding54} is bounded as
	\begin{equation} \label{eq:bounds_ED}
		q_{\mathrm{max}} \leq P_{\E} \leq k q_{\mathrm{max}}.
	\end{equation} 
\label{theorem:bridge}
\end{theorem}
\vspace{-2em}

\begin{proof}
	The block error event is defined as
	\begin{equation*}
		\mathcal{E}_{\E} \triangleq \{(u_1^k, y_1^n)\in\mathcal{X}^k\times\mathcal{Y}^n: \uED_1^k(y_1^n) \neq u_1^k\}
	\end{equation*}
	where $ \uED_1^k(y_1^n) $ is the output of Elias' decoder, obtained by using \eqref{eq:permutation}, \eqref{eq:decision_BEC} and \eqref{eq:metric_calc_ED}. The bit error event is defined as
	\begin{equation*}\label{eq:elias_bit_error_event}
		\mathcal{B}_{i} \triangleq \{(u_1^k, y_1^n)\in\mathcal{X}^k\times\mathcal{Y}^n: \uED_i(y_1^n) \neq u_i\} \quad \text{for}\quad i=1,\ldots,k.
	\end{equation*}
	The block error event satisfies $ \mathcal{E}_{\E} = \cup_{i = 1}^{k} \mathcal{B}_i $, which leads to
	\begin{equation}
		\max_{i = 1,\dots,k} P(\mathcal{B}_i) \leq P_{\E} \leq \sum_{i=1}^{k} P(\mathcal{B}_i).
		\label{eq:bound_MAP}
	\end{equation}
	We conclude the proof by combining \eqref{eq:bound_MAP} and Lemma \ref{lem:worst_channel}.
\end{proof}
\begin{figure}[]
	\begin{center}
		\begin{tikzpicture}[font=\footnotesize]
\begin{semilogyaxis}[
	 mark options={solid,scale=1.2},
     width=0.99*\columnwidth,
     height=0.9*\columnwidth,
     grid=both,
     grid style={dotted,gray!50},
    ylabel near ticks,
    xlabel near ticks,
	xmin=0.1,
	xmax=0.5,
	ymin=1e-3,
	ymax=1,
	xlabel={Erasure probability, $\epsilon$},
	ylabel={Block Error Rate}
	]

\addplot[color=red,line width = 0.3pt,solid,mark=o] table[x=erasure,y=P] {prob_ml.txt}; \label{leg:prob_ml} 
\addplot[color=red,line width = 0.3pt,solid,mark=square] table[x=erasure,y=P] {prob_ed.txt}; \label{leg:prob_ed}  
\addplot[color=red,line width = 0.3pt,solid,mark=star] table[x=erasure,y=P] {prob_scd.txt}; \label{leg:prob_scd}

\addplot[color=black,line width = 0.3pt,solid,mark=diamond] table[x=erasure,y=P] {pub_loosy.txt}; \label{leg:pub_loosy} 

\addplot[color=black,line width = 0.3pt,dashdotted] table[x=erasure,y=P] {pub.txt}; \label{leg:pub} 

\addplot[color=black,line width = 0.3pt,solid,mark=triangle] table[x=erasure,y=P] {plb.txt}; \label{leg:plb}

\addplot[color=black,line width = 0.3pt,solid] table[x=erasure,y=P] {pub_truncated.txt};\label{leg:pub_truncated} 

\end{semilogyaxis}
\end{tikzpicture}
	\end{center}
	\caption{Block error rate vs. erasure probability for the $(9,4)$ \ac{SPC}-\ac{PC} under \ac{ML} \eqref{leg:prob_ml}, Elias' \ref{leg:prob_ed} and \ac{SC} \ref{leg:prob_scd} decoding algorithms. The loose upper bound \eqref{eq:LooseUB_Product} \eqref{leg:pub_loosy}, the upper bound given by the right hand side of \eqref{eq:UB_Product} \eqref{leg:pub} and the lower bound given by the left hand side of \eqref{eq:UB_Product} \eqref{leg:plb} are provided, together with the truncated union bound \eqref{eq:bec_tub} \eqref{leg:pub_truncated}.}\label{fig:FER49}
\end{figure}
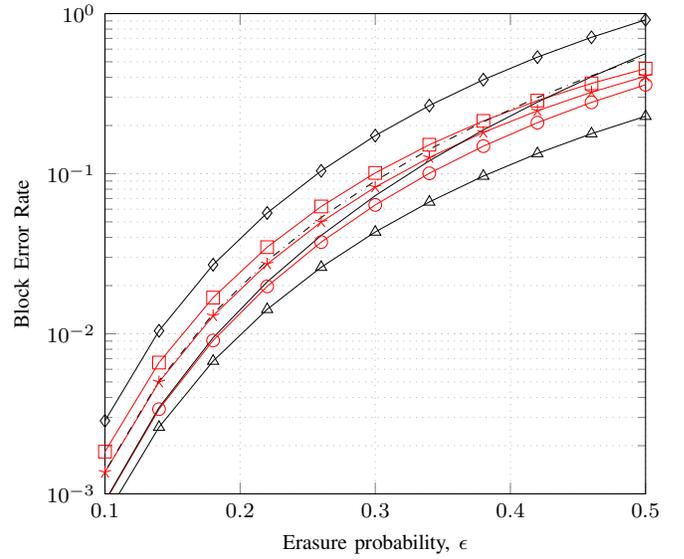

\begin{theorem} For an \ac{SPC}-\ac{PC} over the \ac{BEC}($ \epsilon $),
	\begin{equation}
		P_{\SC} \leq P_{\E}.
	\end{equation}
	\label{theorem:comparison}
\end{theorem}
\vspace{-2em}
\begin{proof}
	Over the \ac{BEC}, both decoders can either make a correct decision or get an erasure for an information bit according to \eqref{eq:decision_BEC}. Therefore, \ac{SC} decoding cannot make use of any wrong bit decision.  Recall \eqref{eq:metric_calc_BEC} for \ac{SC} decoding and \eqref{eq:metric_calc_ED} for Elias' decoding. Under the former, the preceding decisions are exploited. However, each bit decision under the latter is made in the same manner as if one of the past decisions coming from left is an erasure for a local decoder under \ac{SC} decoding. More precisely, assume that we apply both \ac{SC} decoding and Elias' decoding to an observation $y_1^n$. Furthermore, assume at an intermediate step, the \ac{SC} decoder computes the output message $m_{j,i}^{(\ell)}$ corresponding to the $i$-th information bit of the $j$-th local code at the $\ell$-th dimension such that $1\leq\ell\leq m$, $1\leq j\leq\eta_\ell$, $1\leq i\leq n_\ell$. Assume also that there exists at least one $z$, $1\leq z\leq i-1$, such that $\lambda_{j,z}^{(\ell)} = \texttt{e}$. Note that having $\lambda_{j,z}^{(\ell)} = \texttt{e}$ implies also that $\rho_{j,z}^{(\ell)} = 0$. Then, \eqref{eq:metric_calc_BEC} computes the message as $m_{j,i}^{(\ell)} = \rho_{j,i}^{(\ell)}$. For the same scenario under Elias' decoding, \eqref{eq:metric_calc_ED} computes the same message because $\rho_{j,z}^{(\ell)} = 0$ and $\tanh(0) = 0$. Therefore, the probability that they can decode the first bit correctly is the same. Once both decoded the first bit, then the \ac{SC} decoding is at least as good as Elias' decoding.
\end{proof}
\begin{example}
		Consider transmission using the $ (9,4) $ product code with the received vector $ y_1^9 = \{\texttt{e},0,1,0,\texttt{e},\texttt{e},\texttt{e},\texttt{e},1\} $. Under \ac{SC} decoding, the message is decoded correctly while Elias' decoding would fail to decode the 4th information bit.
\end{example}
The bounds and block error probabilities under \ac{SC}, Elias' and \ac{ML} decoding for the $ (9,4) $ \ac{SPC}-\ac{PC} are depicted in Fig. \ref{fig:FER49}. For completeness, we also provide the truncated union bound on the block error probability under \ac{ML} decoding given by
\begin{equation}
	P_B\approx \Amin\epsilon^{\dmin}.
	\label{eq:bec_tub}
\end{equation}
The given error probabilities are not simulation results, but are computed exactly under the three decoding algorithms thanks to the short length of the code. We observe that the bounds are tight especially in the low error rate regime.

\subsection{Performance over the binary input AWGN Channel}
\begin{figure}[]
	\begin{center}
		\begin{tikzpicture}[font=\footnotesize]
\begin{semilogyaxis}[
	 mark options={solid,scale=1.2},
     width=0.99*\columnwidth,
     height=0.9*\columnwidth,
     grid=both,
     grid style={dotted,gray!50},
    ylabel near ticks,
    xlabel near ticks,
	xmin=3,
	xmax=6,
	ymin=1e-6,
	ymax=1,
	xlabel={$E_b/N_0$, [dB]},
	ylabel={Block Error Rate}
	]

\addplot[color=red,line width = 0.3pt,solid,mark=star] table[x=snr,y=P] {216_125_scd.txt}; \label{leg:216_125}
\addplot[color=red,line width = 0.3pt,solid,mark=o] table[x=snr,y=P] {125_64_scd.txt}; \label{leg:125_64}

\addplot[color=black!90,line width = 0.3pt,dashed,mark=star] table[x=snr,y=P] {pub_truncated_216_125.txt}; \label{leg:pub_truncated_216_125}
\addplot[color=black!90,line width = 0.3pt,solid] table[x=snr,y=P] {pub_truncated_125_64.txt}; \label{leg:pub_truncated_125_64} 

\addplot[color=black!90,line width = 0.3pt,solid,mark=star] table[x=snr,y=P] {ed_awgn_216_125.txt}; \label{leg:216_125_ed}
\addplot[color=black!90,line width = 0.3pt,solid,mark=o] table[x=snr,y=P] {ed_awgn_125_64.txt}; \label{leg:125_64_ed}

\end{semilogyaxis}
\end{tikzpicture}
	\end{center}
	\caption{Block error rate vs. signal-to-noise ratio, under \ac{SC} decoding for the $(216,125)$ code \eqref{leg:216_125} and for the $(125,64)$ code \eqref{leg:125_64}, and under Elias' decoding (\ref{leg:216_125_ed} and \ref{leg:125_64_ed}, respectively), compared with the respective truncated union bounds \eqref{eq:awgn_tub} (\ref{leg:pub_truncated_216_125} and \ref{leg:pub_truncated_125_64}).
	}\label{fig:combined}
\end{figure}
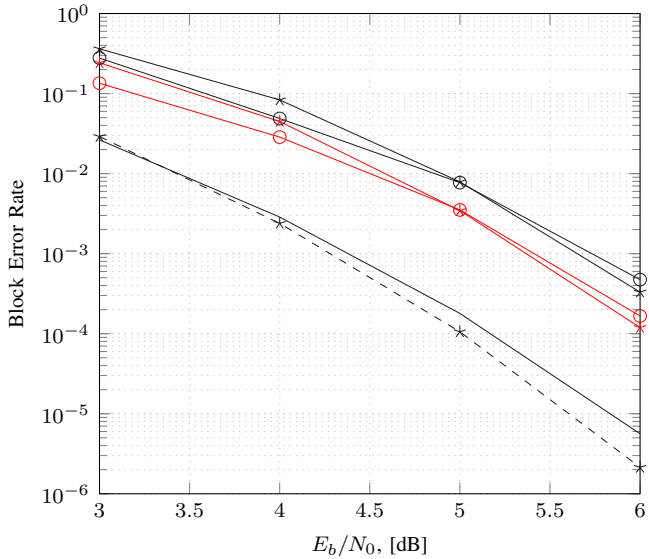

In Fig. \ref{fig:combined}, we provide the block error rate performance of the $3$-dimensional $ (125,64) $ and $ (216,125) $ \ac{SPC}-\acp{PC}, obtained by iterating $(5,4)$ and $(6,5)$ \ac{SPC} codes, respectively, under both \ac{SC} and Elias' decoding algorithms for the \ac{B-AWGN} channel. In the figure, we also show the truncated union bound
\begin{equation}
	P_B\approx \frac{1}{2}\Amin\erfc\sqrt{\dmin R\tfrac{E_b}{N_0}}
	\label{eq:awgn_tub}
\end{equation}
which, for high $E_b/N_0$, approximates tightly the \ac{ML} decoding performance. Here, $E_b$ denotes the energy per information bit and $N_0$ is the single-sided noise power spectral density. For both codes, the difference of \ac{SC} decoding to the truncated union bound is around $ 1 $ dB at block error rate of $ 10^{-4} $. In addition, \ac{SC} decoding outperforms Elias' decoding.

	\section{Conclusions}

We introduced successive cancellation decoding of the class of product codes obtained by iterating single parity-check codes. Thanks to the structure of \ac{SPC}-\acp{PC}, we showed how to compute the decision metrics recursively under \ac{SC} decoding, which resembles \ac{SC} decoding of polar codes. In addition, we introduced an analysis on the binary erasure channel, yielding lower and upper bounds on the block error probability. The performance of \ac{SC} decoding is then compared with that of the original decoder of \acp{PC} introduced by Elias, showing that \ac{SC} decoding yields lower block error probability than Elias' decoding. Finally, it is concluded for the analyzed codes that the low-complexity \ac{SC} decoding can outperform Elias' decoding by $\sim 0.35$ dB.



\end{document}